%% file: main.tex
\documentclass[11pt]{article}
\usepackage[a4paper,margin=1in]{geometry}
\IfFileExists{glyphtounicode.tex}{\input{glyphtounicode}\pdfgentounicode=1}{}
\usepackage{cmap}
\usepackage[T1]{fontenc}
\usepackage[utf8]{inputenc}
\usepackage{lmodern}
\usepackage{microtype}
\usepackage{amsmath,amssymb,amsthm,mathtools}
\usepackage{enumitem}
\usepackage[numbers,sort&compress]{natbib}
\usepackage{hyperref}
\hypersetup{colorlinks=true,linkcolor=blue,citecolor=blue,urlcolor=blue,
  pdftitle={Biprofile Deviation Logic: Report-Replacement Frames and Audit Witnesses},
  pdfauthor={Faruk Alpay and Baris Basaran},
  pdfkeywords={coalitional modal logic, strategic social choice, strategy-proofness, report-deviation frames, canonical completeness, single-peaked preferences}}
\allowdisplaybreaks
\newtheorem{theorem}{Theorem}
\newtheorem{proposition}{Proposition}
\newtheorem{definition}{Definition}
\newtheorem{lemma}{Lemma}
\newtheorem{corollary}{Corollary}
\newtheorem{remark}{Remark}

\begin{document}

\begin{center}
{\small \textbf{Article}}\\[1.2em]
{\LARGE\bfseries Biprofile Deviation Logic: Report-Replacement Frames and Audit Witnesses}\\[1.0em]
Faruk Alpay\textsuperscript{1,*}\quad and\quad Baris Ba\c{s}aran\textsuperscript{1}\\[0.6em]
{\small \textsuperscript{1} Department of Computer Engineering, Bahcesehir University, Istanbul, Turkey;\newline
\texttt{faruk.alpay@bahcesehir.edu.tr} (F.A.); \texttt{baris.basaran@bahcesehir.edu.tr} (B.B.)}\\[0.4em]
{\small \textsuperscript{*} Correspondence: \texttt{alpay@lightcap.ai}. Institutional e-mail: \texttt{faruk.alpay@bahcesehir.edu.tr}.}
\end{center}

\vspace{0.5em}
\noindent\rule{\textwidth}{0.4pt}
\vspace{0.5em}

\input{body.tex}

\input{references_manual.tex}
\end{document}

%% file: body.tex
\begin{abstract}
Biprofile deviation logic models strategic social choice states as pairs \((R,P)\), where \(R\) is the true profile used for welfare comparisons and \(P\) is the submitted report profile used by the rule. Coalition modalities replace only the reports of the coalition, and their relations satisfy the fixed law \(E_C\circ E_D=E_{C\cup D}\). The paper proves soundness and completeness of \(H_{\mathrm{bp}}\) for the abstract frame class \(\mathsf{Dev}(N)\), with the reverse-composition midpoint displayed inside the canonical proof. It then separates abstract \(\mathsf{Dev}(N)\)-components from genuine report-coordinate products by coordinate separation. On the social-choice side, the classical facts supply the source notions; the paper-specific contribution is the audit layer for representation changes: typed manipulation witnesses, a boundary-row theorem for off-domain extensions, and a factor-closure criterion for public deletions. The ancillary material contains the input formats, an executable certificate checker, Lean and Alloy companions for the finite relational lemmas and update patterns, recorded run logs, and checksums.
\end{abstract}


\section{Introduction}

Strategic social choice compares two profiles. A true profile determines welfare comparisons, while a submitted report profile determines the chosen outcome. Biprofile deviation logic makes this split part of the Kripke state. A state is a pair \((R,P)\), with \(R\) the true profile and \(P\) the submitted profile, and a coalition modality changes only the report coordinates of the coalition:
\[
   (R,P)\ E_C\ (R,Q).
\]
The edge records a possible report change; the welfare comparison used to judge the change is still read from \(R\).

The proof-theoretic object is the abstract algebra of these labelled report-replacement edges. For coalitions \(C,D\subseteq N\), two successive replacements compose by union:
\[
   E_C\circ E_D=E_{C\cup D}.
\]
Together with S5 behaviour for each label, empty-coalition identity, and monotonicity under inclusion, this defines the frame class \(\mathsf{Dev}(N)\). The central theorem proves soundness and completeness of a Hilbert system \(H_{\mathrm{bp}}\) for that class.

The modal-logic claim is deliberately specific, and its status is separated from the general correspondence background. General Sahlqvist/canonicity and Kracht-style results account for the canonicity and first-order correspondence of the fixed S5, identity, inclusion, and reduction axioms once the language contains one primitive modality for each coalition \cite{Sahlqvist1975,BlackburnDeRijkeVenema2001,Kracht1999}. That is the background part: it explains why the axiom table is canonical. The paper-specific part starts from the interpretation and use of that table. The labels are not arbitrary S5 modalities; they are the report-replacement actions of the fixed semilattice \((2^N,\cup,\emptyset)\), where composition is coalition union and truth is split between true preferences and submitted reports. Theorem~\ref{thm:completeness} gives a self-contained proof for this fixed table and exposes the midpoint step \(R_{C\cup D}\subseteq R_C\circ R_D\), the step that makes exact report factorization visible inside the canonical construction. Proposition~\ref{prop:rectangular-representation} is independent of that metatheorem route: it proves that the abstract equations do not by themselves recover report coordinates, and identifies the extra coordinate-separation condition needed for a genuine product representation.

The paper has three main contributions.
\begin{description}[leftmargin=7mm,style=nextline]
  \item[Modal core.] Theorem~\ref{thm:completeness} fixes the coalition signature and proves completeness for the concrete multiplication table \(C,D\mapsto C\cup D\). The standard canonicity literature supplies the general preservation background; the paper supplies the report-deviation reading and the direct midpoint proof for the reverse composition row.
  \item[Representation boundary.] Proposition~\ref{prop:rectangular-representation} is not another completeness theorem. It shows where product logic would start with coordinates and where \(\mathsf{Dev}(N)\) does not: coordinate separation is exactly what turns an abstract component into a report-coordinate product.
  \item[Audit theorems and case study.] Strategy-proofness, group strategy-proofness, Gibbard--Satterthwaite, and single-peaked median results are translated, not reproved. The new social-choice statements are not replacement impossibility theorems; they are delta theorems for changing a representation after an intended domain has already been certified. Theorem~\ref{thm:boundary-safety} says that a larger report domain can create only boundary-row witnesses. Proposition~\ref{prop:update-safety} says that a public deletion preserves the \(\mathsf{Dev}(N)\) algebra exactly when surviving edges have surviving factor midpoints. Section~\ref{sec:single-peaked} runs both tests on a five-agent single-peaked median example, reporting the same witness as \textsc{edge-deleted}, \textsc{boundary-witness}, or \textsc{unsafe-update} under three representation changes.
\end{description}

\noindent\textbf{Restricted-domain caution.}
Manipulation formulas can be evaluated on any displayed labelled audit graph. The completeness theorem applies only to frames satisfying the \(\mathsf{Dev}(N)\) equations. Coordinate-product report domains, such as \(\prod_i M_i\), preserve those equations. An arbitrary subset of report profiles need not: a mixed report required for an \(E_C\circ E_D\) factorization may be missing. Public restrictions are therefore checked by the factor-closure criterion of Proposition~\ref{prop:update-safety}.

\noindent\textbf{Language and input conventions.}
The completeness theorem is stated for the pure coalition-modal language with arbitrary propositional letters. The social-choice application interprets letters as outcome atoms \(o_x\) and true-preference atoms \(p^i_{xy}\). In Section~\ref{sec:complexity}, finite claims are input-relative: a rule may be listed as a table, generated by a domain oracle, submitted as an explicit labelled graph, or audited as a typed modal certificate. The fixed-signature results assume that \(N\) is fixed before the input is read.

\smallskip
\noindent\textbf{Notation.}
Throughout, \(R\) denotes the true-preference profile and \(P,Q\) denote submitted report profiles. The atom \(o_x\) says that the submitted report profile selects alternative \(x\), while \(p^i_{xy}\) says that agent \(i\) truly ranks \(x\) above \(y\). The relation \(E_C\) is the report-deviation relation for coalition \(C\). The class \(\mathsf{Dev}(N)\) consists of the labelled S5 frames satisfying empty-coalition identity, coalition inclusion, and exact union-composition, and \(H_{\mathrm{bp}}\) is the Hilbert system for this class.

\begin{figure}[t]
\centering
\setlength{\fboxsep}{6pt}
\fbox{\begin{minipage}{0.92\linewidth}\small
\textbf{Proof and audit dependency map.}
\[
\begin{array}{ccccc}
\text{Report algebra} & \Longrightarrow & \mathsf{Dev}(N) & \Longrightarrow & H_{\mathrm{bp}}\text{ completeness}\\
\text{Lemma~\ref{lem:algebra}} && \text{D1--D4} && \text{Lemma~\ref{lem:canonical-factorization}, Theorem~\ref{thm:completeness}}\\[1mm]
\mathsf{Dev}(N) & + & \text{coordinate separation} & \Longrightarrow & \text{report product recovery}\\
\text{Lemma~\ref{lem:rectangular-mixing}} && \text{Proposition~\ref{prop:rectangular-representation}} && \\[1mm]
\text{witness record} & + & \text{domain change} & \Longrightarrow & \text{boundary-row audit}\\
\text{Lemma~\ref{lem:witness}} && \text{Theorem~\ref{thm:boundary-safety}} && \\[1mm]
\text{public survivor set} & + & \text{factor closure} & \Longleftrightarrow & \mathsf{Dev}(N)\text{ subframe}
\end{array}
\]
\end{minipage}}
\caption{Dependency map for the main formal claims. The top row is the modal proof-theoretic spine; the lower rows are the representation and audit uses of the same relation algebra.}
\label{fig:dependency-map}
\end{figure}

\paragraph{Comparison with nearby formalisms.}
The closest modal formalisms differ in their primitives and in what they assume at the start. A fusion of S5 logics gives independent equivalence relations, but not the coalition order, the identity action for \(\emptyset\), or the multiplication table \(C\cdot D=C\cup D\). Product logics start from visible dimensions and study the modal logic of the resulting product frame \cite{GabbayShehtman1998,GabbayShehtman2002,Shehtman2023}; \(\mathsf{Dev}(N)\) starts with an abstract semilattice action and then proves when coordinates can be recovered. Dynamic logic and relation algebra expose composition as program or relation syntax \cite{HarelKozenTiuryn2000,Balbiani2001,Maddux2006}; here composition is not an object-language term former, but a finite table of primitive coalition modalities. Polymodal correspondence theory supplies the canonicity background for that table. What it does not supply is the report-deviation interpretation of the labels, the theorem separating abstract components from coordinate products, or the two audit tests that diagnose changes of representation: boundary rows for off-domain rule extensions and factor closure for public deletions.

\paragraph{Fixed-signature separation.}
For fixed \(N\), the language is an ordinary finite polymodal language with one primitive modality for each coalition. Standard canonicity supplies the background guarantee that the S5, identity, inclusion, and reduction principles define a canonical elementary class. The paper-specific part begins after that guarantee: Theorem~\ref{thm:completeness} packages the fixed multiplication row \(C,D\mapsto C\cup D\) as report replacement and gives the explicit midpoint proof for the reverse inclusion \(R_{C\cup D}\subseteq R_C\circ R_D\); Proposition~\ref{prop:rectangular-representation} then proves that satisfying the abstract equations is still weaker than being a genuine report-coordinate product. Thus the completeness theorem is the fixed-signature proof-theoretic spine, while the representation theorem is the boundary result that product logic would assume at the outset.

Coalition logic and game-form semantics provide languages for ability and effectivity \cite{Pauly2002,PelegPeters2010,PelegHolzman2017}. Modal social-choice work has represented social-choice functions and classical theorems \cite{TroquardHoekWooldridge2011,CinaEndriss2016}, and dynamic or epistemic approaches study strategic voting and mechanism change \cite{ChopraPacuitParikh2004,vanDitmarschLangSaffidine2013,ParmannAgotnes2021,MittelmannEtAl2025,GalimullinEtAl2025}. The present framework is narrower: its modalities are report-replacement edges, its atoms distinguish true preferences from reported outcomes, and its audit results track which field of a stored witness changes when a report domain or public submodel is altered.

\noindent\textbf{Organization.}
Section~2 defines biprofile semantics. Section~3 proves soundness and completeness for abstract \(\mathsf{Dev}(N)\)-frames and separates product representability by coordinate separation. Sections~4--6 give the social-choice translations and witness language. Section~7 gives the single-peaked/off-domain/update example. Section~8 records the finite audit contract, an explicit supplied-certificate schema, a small Lean proof companion, bounded Alloy checks, and enumeration baselines.

\section{Biprofile semantics}

Let \(N=\{1,\ldots,n\}\) be a finite set of agents and let \(A\) be a finite set of alternatives, \(|A|=m\geq 2\). Write \(\mathcal L(A)\) for the set of strict linear orders on \(A\), and put \(L=|\mathcal L(A)|=m!\). A profile is an element of \(\mathcal L(A)^N\). If \(R\in\mathcal L(A)^N\), then \(x\succ_i^R y\) means that agent \(i\) ranks \(x\) above \(y\) in the true order \(R_i\), and \(x\succeq_i^R y\) abbreviates \(x=y\) or \(x\succ_i^R y\).

A resolute social choice function is a map
\[
   f:\mathcal L(A)^N\to A .
\]
It induces a normal-form game form whose pure strategies are reported orders and whose outcome at report profile \(P\) is \(f(P)\). The associated biprofile state space is
\[
   S_f=\mathcal L(A)^N\times\mathcal L(A)^N.
\]
The sincere states are
\[
   T_f=\{(R,R):R\in\mathcal L(A)^N\}.
\]
For \(C\subseteq N\), write \(Q\equiv_{-C}P\) if \(Q_j=P_j\) for every \(j\notin C\).

\begin{definition}[Biprofile model]
The biprofile model associated with \(f\) is
\[
   \mathfrak M_f=(S_f,\{E_C:C\subseteq N\},V_f),
\]
where
\[
  (R,P)E_C(R',Q)
  \quad\text{iff}\quad
  R'=R\ \text{ and }\ Q\equiv_{-C}P.
\]
The valuation is defined by
\[
\begin{array}{rcl}
  \mathfrak M_f,(R,P)\models o_x
  &\text{iff}& f(P)=x,\\[1mm]
  \mathfrak M_f,(R,P)\models p^i_{xy}
  &\text{iff}& x\succ_i^R y,\\[1mm]
  \mathfrak M_f,(R,P)\models t^i_x
  &\text{iff}& x \text{ is the } R_i\text{-maximal alternative.}
\end{array}
\]
Here \(o_x\) is an outcome atom, \(p^i_{xy}\) is a true-preference atom, and \(t^i_x\) is a true-top atom.
\end{definition}

The language \(\mathcal L_{\mathrm{bp}}(N,A)\) is generated by
\[
  \varphi ::= o_x \mid p^i_{xy} \mid t^i_x \mid
  \neg\varphi \mid (\varphi\wedge\varphi) \mid \langle C\rangle\varphi,
\]
where \(x,y\in A\), \(i\in N\), and \(C\subseteq N\). The dual \([C]\varphi\) abbreviates \(\neg\langle C\rangle\neg\varphi\). The modal clause is
\[
  \mathfrak M_f,(R,P)\models \langle C\rangle\varphi
  \quad\text{iff}\quad
  \text{there is }Q\equiv_{-C}P\text{ such that }
  \mathfrak M_f,(R,Q)\models\varphi .
\]
Thus \(\langle C\rangle\) is a report-changing modality. It never changes \(R\). This invariance is the semantic reason why a formula can compare an outcome after a deviation with preferences held fixed before the deviation.

For orientation, consider two agents and a current report profile \(P=(P_1,P_2)\) under a fixed true profile \(R\). Agent 1's deviation relation moves horizontally by replacing only \(P_1\); agent 2's deviation relation moves vertically by replacing only \(P_2\); the grand-coalition relation allows both moves. The logical diamond \(\langle C\rangle\) therefore behaves like a controlled coordinate projection rather than like an arbitrary graph edge.

\begin{figure}[!htbp]
\centering
\[
\begin{array}{ccc}
(R,(P_1,P_2)) & \xrightarrow{E_{\{1\}}} & (R,(Q_1,P_2))\\[2mm]
\downarrow E_{\{2\}} & & \downarrow E_{\{2\}}\\[2mm]
(R,(P_1,Q_2)) & \xrightarrow{E_{\{1\}}} & (R,(Q_1,Q_2))
\end{array}
\qquad
(R,(P_1,P_2))\xrightarrow{E_{\{1,2\}}}(R,(Q_1,Q_2))
\]
\caption{Coordinate reading of report deviations. All states share the same true profile \(R\); only submitted reports vary.}
\label{fig:biprofile-coordinate}
\end{figure}

\subsection*{Running example: plurality with one profitable report change}
We use the following small example only as an orientation device; none of the later theorems depends on plurality. Let \(N=\{1,2,3\}\), let \(A=\{a,b,c\}\), and let \(f\) be plurality with the fixed tie-breaking order \(a\triangleright b\triangleright c\). Take the true profile
\[
   R_1:b\succ a\succ c,
   \qquad
   R_2:a\succ b\succ c,
   \qquad
   R_3:c\succ b\succ a .
\]
At the sincere state \(s_0=(R,R)\), the top votes are one for each alternative, so the tie-break selects \(a\); hence \(o_a\) holds at \(s_0\). Now let agent 3 submit the report \(Q_3:b\succ a\succ c\), while agents 1 and 2 keep their sincere reports. Put \(Q=(R_1,R_2,Q_3)\). Then \(s_0E_{\{3\}}(R,Q)\), because only the third report coordinate has changed, and plurality now selects \(b\). Since agent 3's true order is \(c\succ b\succ a\), the atom \(p^3_{ba}\) is true at both states. Thus the witness for \(\mu_{\{3\}}\) is the successor \((R,Q)\): the current outcome is \(a\), an \(E_{\{3\}}\)-neighbour has outcome \(b\), and the true-coordinate says that agent 3 prefers \(b\) to \(a\).
 In modal notation this concrete state verifies
\[
   \mathfrak M_f,s_0\models
   o_a\wedge \langle\{3\}\rangle(o_b\wedge p^3_{ba}),
\]
so the one-agent strategy-proofness sentence fails for this particular plurality rule. This is the small case study used throughout the paper: the social-choice fact is elementary, but the example makes visible which coordinate each symbol reads and which coordinate a deviation is allowed to mutate.

The same example also shows the coalition algebra. If agent 1 changes a report and then agent 3 changes a report, the resulting state differs from \(s_0\) only in coordinates \(\{1,3\}\). Lemma~\ref{lem:algebra} says that, at the level of one-shot reachability, this two-step path is exactly an \(E_{\{1,3\}}\)-move. The S5 part of the logic records that each fixed coalition can move back and forth inside its own report fibre; the union axiom records that independent coordinate changes compose by taking the union of the changed indices.

This running example fixes the interpretation of the later formal clauses. A modal atom such as \(o_b\) is the result of evaluating the rule on the report coordinate. A welfare atom such as \(p^3_{ba}\) is read from the true coordinate. A manipulation witness is therefore a labelled edge in the report graph together with a comparison in the fixed true profile. The later proof theory abstracts exactly this labelled graph structure, while the computational sections ask how such witnesses are generated or represented under different input models.

\paragraph{A worked witness record.}
For the plurality state above, the witness can be written as
\[
   w=(R,R,\{3\},Q,a,b),
\]
where \((R,R)E_{\{3\}}(R,Q)\), \(f(R)=a\), \(f(Q)=b\), and \(p^3_{ba}\) is read from the unchanged true coordinate. This record is more than a restatement of manipulation because each component can be varied independently. If a restricted report domain removes \(Q_3\), the same welfare comparison and rule values remain available as data, but the edge field of \(w\) is absent. If an off-domain extension admits \(Q_3\), the formula \(\mu_{\{3\}}\) is unchanged and the failure is localized as an added successor. If a report update overwrites agent 3's report by \(Q_3\), the checking question is whether the resulting public image still has the factor closure required by the \(\mathsf{Dev}(N)\) equations. Thus the witness separates four possible sources of failure: the true-preference comparison, the current rule value, the deviated rule value, and the admissible report edge. Section~\ref{sec:single-peaked} uses the same record format for the restricted-domain/off-domain/update case study.

\begin{lemma}[Algebra of report-deviation relations]\label{lem:algebra}
For every finite resolute rule \(f\) and all coalitions \(C,D\subseteq N\), the following hold on each true-profile fibre \(\{R\}\times\mathcal L(A)^N\):
\begin{enumerate}[label=(\alph*)]
  \item \(E_C\) is an equivalence relation.
  \item If \(C\subseteq D\), then \(E_C\subseteq E_D\).
  \item \(E_\varnothing\) is the identity relation, and \(E_N\) is the universal relation on the fibre.
  \item \(E_C\circ E_D=E_{C\cup D}\).
\end{enumerate}
Consequently,
\[
  \mathfrak M_f,s\models \langle C\rangle\langle D\rangle\varphi
  \quad\text{iff}\quad
  \mathfrak M_f,s\models \langle C\cup D\rangle\varphi .
\]
\end{lemma}

\begin{proof}
Reflexivity, symmetry, and transitivity of \(E_C\) are inherited from equality of outsider reports. Inclusion is immediate. The empty coalition changes no coordinate, while the grand coalition may change every report coordinate. For composition, if \((R,P)E_C(R,Q)E_D(R,H)\), then every coordinate on which \(P\) and \(H\) differ is in \(C\cup D\). Conversely, if \((R,P)E_{C\cup D}(R,H)\), define
\[
Q_i=
\begin{cases}
H_i, & i\in C\setminus D,\\
P_i, & \text{otherwise.}
\end{cases}
\]
Then \((R,P)E_C(R,Q)\) and \((R,Q)E_D(R,H)\). The final equivalence is the modal reading of relational composition.
\end{proof}

This elementary algebra is the reason for using coalition labels instead of a single untyped deviation relation. The formula \(\langle\{1\}\rangle\langle\{2\}\rangle\varphi\) says that agent 1 first changes her report and agent 2 then changes his report. Lemma~\ref{lem:algebra} states that, for one-shot reachability, the order of these coordinate changes is immaterial and the composite move is exactly an \(E_{\{1,2\}}\)-move. Later dynamic updates refine this picture by adding preconditions and postconditions to such coordinate changes.

\section{Proof theory of deviation frames}\label{sec:proof}

Let \(\mathsf{Dev}(N)\) be the class of frames \((S,\{E_C\}_{C\subseteq N})\) satisfying the abstract report-deviation equations below:
\begin{enumerate}[label=(D\arabic*),leftmargin=8mm]
  \item each \(E_C\) is an equivalence relation;
  \item \(E_\emptyset=\{(s,s):s\in S\}\) is the identity relation;
  \item if \(C\subseteq D\), then \(E_C\subseteq E_D\);
  \item \(E_C\circ E_D=E_{C\cup D}\) for all \(C,D\subseteq N\).
\end{enumerate}

The concrete grand-coalition relation is universal inside each fixed true-profile fibre. Abstractly, \(E_N\)-components are the fibres on which all report coordinates may vary; no global universal relation across components is imposed. Every concrete biprofile frame belongs to \(\mathsf{Dev}(N)\), but an arbitrary \(\mathsf{Dev}(N)\)-component need not be a coordinate product. Theorem~\ref{thm:completeness} is therefore an abstract completeness theorem. The next lemmas record what is automatic from the \(\mathsf{Dev}(N)\) equations and what must be added to recover genuine coordinates.

For a proof-theory reader, the important distinction is between correspondence and construction. The modal axiom \(\langle C\rangle\langle D\rangle\varphi\leftrightarrow\langle C\cup D\rangle\varphi\) corresponds semantically to \(E_C\circ E_D=E_{C\cup D}\). Completeness still requires the canonical relations to realize this equality. The forward half is immediate from the box translation, but the reverse half requires constructing a canonical midpoint. Lemma~\ref{lem:canonical-factorization} isolates exactly that step; no product coordinates or external program algebra are used in it.

A useful way to view exact composition is as a finite rectangular-mixing property. Given two states \(s,t\) in the same \(E_N\)-component and a coalition \(C\), the equation \(E_C\circ E_{N\setminus C}=E_N\) forces an intermediate state that combines the \(C\)-move and the outside-\(C\) move. Thus missing corners cannot occur inside a \(\mathsf{Dev}(N)\)-frame. They can occur only after an operation, such as a public deletion, is treated as a submodel without rechecking the union-composition equations.

\begin{lemma}[Rectangular mixing from exact composition]
\label{lem:rectangular-mixing}
Let \((S,\{E_C\}_{C\subseteq N})\in\mathsf{Dev}(N)\). If \(sE_Nt\), then for every \(C\subseteq N\) there is a state \(u\) such that \(uE_Cs\) and \(uE_{N\setminus C}t\).
\end{lemma}

\begin{proof}
Since \(E_C\circ E_{N\setminus C}=E_N\), the edge \(sE_Nt\) gives a state \(u\) with \(sE_Cu\) and \(uE_{N\setminus C}t\). The relation \(E_C\) is symmetric, so \(uE_Cs\).
\end{proof}

\begin{proposition}[Finite coordinate representation]
\label{prop:rectangular-representation}
Let \(X\) be an \(E_N\)-component of a finite abstract deviation frame. Then \(X\) is isomorphic to a concrete report-coordinate product \(\prod_{i\in N}M_i\), with
\[
   s\,E_C\,t\quad\text{iff}\quad h_i(s)=h_i(t)\text{ for all }i\notin C,
\]
if and only if it satisfies coordinate separation:
\[
   \bigl(s\,E_{N\setminus\{i\}}\,t\text{ for every }i\in N\bigr)
   \quad\Longrightarrow\quad s=t .
\]
\end{proposition}

\begin{proof}
Necessity is immediate for a concrete product: if two tuples agree in every coordinate, they are the same tuple.

For sufficiency, define the candidate coordinate set for agent \(i\) by
\[
   M_i=X/E_{N\setminus\{i\}},
\]
and define
\[
   h:X\to\prod_{i\in N}M_i,
   \qquad
   h(s)=\big([s]_{E_{N\setminus\{1\}}},\ldots,[s]_{E_{N\setminus\{n\}}}\big).
\]
If \(h(s)=h(t)\), then \(sE_{N\setminus\{i\}}t\) for every \(i\). Coordinate separation gives \(s=t\), so \(h\) is injective.

To prove surjectivity, take an arbitrary tuple \((m_1,\ldots,m_n)\in\prod_iM_i\), and choose representatives \(s_i\in X\) with \(m_i=[s_i]_{E_{N\setminus\{i\}}}\). We construct states \(u_k\) whose first \(k\) coordinates are the target classes. Start with any \(u_0\in X\). Given \(u_k\), apply Lemma~\ref{lem:rectangular-mixing} inside the \(E_N\)-component to the pair \((u_k,s_{k+1})\) with coalition \(\{k+1\}\). The resulting state \(u_{k+1}\) satisfies
\[
   u_{k+1}E_{\{k+1\}}u_k
   \quad\text{and}\quad
   u_{k+1}E_{N\setminus\{k+1\}}s_{k+1}.
\]
The first relation preserves the first \(k\) already fixed coordinate classes, and the second fixes the \((k+1)\)-st coordinate class. After \(n\) steps, \(h(u_n)=(m_1,\ldots,m_n)\).

It remains to prove that \(h\) preserves and reflects the labelled relations. If \(sE_Ct\), then by monotonicity \(sE_{N\setminus\{i\}}t\) for each \(i\notin C\), because \(C\subseteq N\setminus\{i\}\). Hence \(h_i(s)=h_i(t)\) outside \(C\).

Conversely, suppose \(h_i(s)=h_i(t)\) for every \(i\notin C\). Lemma~\ref{lem:rectangular-mixing} gives \(w\in X\) such that \(wE_Cs\) and \(wE_{N\setminus C}t\). If \(i\notin C\), then \(wE_Cs\), monotonicity, and the assumption \(sE_{N\setminus\{i\}}t\) imply \(wE_{N\setminus\{i\}}t\). If \(i\in C\), then \(N\setminus C\subseteq N\setminus\{i\}\), so \(wE_{N\setminus C}t\) implies \(wE_{N\setminus\{i\}}t\). Thus \(wE_{N\setminus\{i\}}t\) for every \(i\), and coordinate separation gives \(w=t\). Since \(wE_Cs\), symmetry yields \(sE_Ct\). Therefore \(h\) is a labelled-frame isomorphism onto the concrete product.
\end{proof}

The representation boundary is therefore not rectangularity itself. Exact union-composition already supplies the necessary mixed corners inside \(\mathsf{Dev}(N)\). The extra condition needed for genuine report coordinates is coordinate separation: different abstract states must not be indistinguishable by all single-coordinate quotients. This separation is useful because the proof theory concerns the abstract algebra, whereas concrete social-choice models supply separated report fibres.

\paragraph{Two finite frames for orientation.}
For two agents, the concrete two-by-two report square is the product case. Its four states may be written \((0,0),(1,0),(0,1),(1,1)\). The relation \(E_{\{1\}}\) keeps the second coordinate fixed, \(E_{\{2\}}\) keeps the first coordinate fixed, \(E_\emptyset\) is identity, and \(E_{\{1\}}\circ E_{\{2\}}=E_{\{1,2\}}\) is the full square inside the component.

The abstract class is larger. Let \(X=\{a,b\}\), let \(E_\emptyset\) be identity, and let \(E_{\{1\}}=E_{\{2\}}=E_{\{1,2\}}=X\times X\). Then all \(\mathsf{Dev}(N)\) equations hold: both singleton relations are equivalences, inclusion is trivial, and every non-empty composition is universal. The component is not a coordinate product in the intended separated sense. Coordinate separation fails, because \(aE_{\{1\}}b\) and \(aE_{\{2\}}b\) while \(a\ne b\). This example is why Theorem~\ref{thm:completeness} is stated for abstract frames, and why Proposition~\ref{prop:rectangular-representation} is needed separately for concrete report products.

\begin{remark}[Topological reading only]
Each \(E_C\)-class may be viewed as a block of a partition topology, but no topological completeness theorem is used. The formal preservation result below is the canonical completeness theorem.
\end{remark}

Let \(\mathsf H_{\mathrm{bp}}\) be the smallest normal multimodal Hilbert system containing all propositional tautologies, modus ponens, necessitation for each \([C]\), and the following schemes:
\[
\begin{array}{ll}
\mathsf K_C: & [C](\varphi\to\psi)\to([C]\varphi\to[C]\psi),\\[1mm]
\mathsf T_C: & [C]\varphi\to\varphi,\\[1mm]
\mathsf I_\emptyset: & [\emptyset]\varphi\leftrightarrow\varphi,\\[1mm]
\mathsf 4_C: & [C]\varphi\to[C][C]\varphi,\\[1mm]
\mathsf 5_C: & \neg[C]\varphi\to[C]\neg[C]\varphi,\\[1mm]
\mathsf M_{CD}: & \langle C\rangle\varphi\to\langle D\rangle\varphi\quad(C\subseteq D),\\[1mm]
\mathsf U_{CD}: & \langle C\rangle\langle D\rangle\varphi\leftrightarrow
\langle C\cup D\rangle\varphi .
\end{array}
\]
The following lemma records the exact dual forms used in the canonical proof.

\begin{lemma}[Dual forms of inclusion and union]
\label{lem:dual-forms}
In any normal modal logic with the displayed diamond schemes, the following box schemes are derivable:
\[
   \mathsf M_{CD}^{\Box}:\quad [D]\chi\to[C]\chi\quad(C\subseteq D),
   \qquad
   \mathsf U_{CD}^{\Box}:\quad [C\cup D]\chi \leftrightarrow [C][D]\chi .
\]
Conversely, the displayed diamond schemes are obtained from these box schemes by duality.
\end{lemma}

\begin{proof}
For inclusion, dualize \(\langle C\rangle\neg\chi\to\langle D\rangle\neg\chi\) to obtain \([D]\chi\to[C]\chi\). For union, substitute \(\neg\chi\) into \(\langle C\rangle\langle D\rangle\varphi\leftrightarrow\langle C\cup D\rangle\varphi\), negate both sides, and use \([E]\chi:=\neg\langle E\rangle\neg\chi\). This gives
\[
   \neg\langle C\rangle\langle D\rangle\neg\chi
   \leftrightarrow
   \neg\langle C\cup D\rangle\neg\chi,
\]
which is \([C][D]\chi\leftrightarrow[C\cup D]\chi\). The converse direction is the same calculation with boxes and diamonds interchanged.
\end{proof}

The explicit identity scheme \(\mathsf I_\emptyset\) is essential: without it the canonical relation \(R_\emptyset\) would be an S5 equivalence relation, but not necessarily the identity relation required by \(\mathsf{Dev}(N)\).

The proof is canonical, but the non-routine obligation is worth isolating. By Lemma~\ref{lem:dual-forms}, the box form of \(\mathsf U_{CD}\) propagates boxed formulas along an \(R_C\)-edge followed by an \(R_D\)-edge, giving \(R_C\circ R_D\subseteq R_{C\cup D}\). The converse asks for a midpoint: from \(XR_{C\cup D}Z\) one must find a maximal theory \(Y\) with \(XR_CY\) and \(YR_DZ\). The proof below obtains \(Y\) from a finite-consistency argument. This keeps exact composition inside the canonical frame construction and does not presuppose product coordinates. Concrete product representation is addressed separately in Proposition~\ref{prop:rectangular-representation}.

\begin{figure}[!htbp]
\centering
\fbox{\begin{minipage}{0.90\textwidth}
\centering
\(\text{S5 labels}\) \(\Rightarrow\) \(R_C\) equivalence\quad\quad
\(\mathsf I_\emptyset\) \(\Rightarrow\) \(R_\emptyset=\mathrm{Id}\)\\[1mm]
\(\mathsf M_{CD}\) \(\Rightarrow\) \(R_C\subseteq R_D\) for \(C\subseteq D\)\\[1mm]
\(\mathsf U_{CD}^{\Box}\) \(\Rightarrow\) \(R_C\circ R_D\subseteq R_{C\cup D}\)\\[1mm]
\(XR_{C\cup D}Z\) \(+\) consistency of \(\Gamma\) \(\Rightarrow\)
\(\exists Y\,(XR_CY\text{ and }YR_DZ)\).
\end{minipage}}
\caption{Proof roadmap for the canonical frame argument. The only step not inherited from the routine S5/inclusion checks is the midpoint construction for the reverse composition inclusion.}
\label{fig:canonical-roadmap}
\end{figure}

\begin{lemma}[Canonical frame and factorization]
\label{lem:canonical-factorization}
Let \(\mathfrak M^c=(W,\{R_C\}_{C\subseteq N},V)\) be the canonical model of \(\mathsf H_{\mathrm{bp}}\), where
\[
   X R_C Y
   \quad\text{iff}\quad
   \{\psi:[C]\psi\in X\}\subseteq Y .
\]
Then \((W,\{R_C\})\in\mathsf{Dev}(N)\). Moreover, for all coalitions \(C,D\),
\[
   R_C\circ R_D=R_{C\cup D},
\]
and if \(XR_CY\), then \([C]\chi\in X\) iff \([C]\chi\in Y\).
\end{lemma}

\begin{proof}
We verify the four frame clauses of \(\mathsf{Dev}(N)\) and then the auxiliary invariance claim.

\emph{S5 components.} For each fixed coalition \(C\), the schemes \(\mathsf T_C,\mathsf 4_C,\mathsf 5_C\), together with normality, give the usual canonical S5 argument. Hence \(R_C\) is reflexive, transitive, and Euclidean, and therefore an equivalence relation.

\emph{Empty coalition.} The scheme \(\mathsf I_\emptyset\) gives \([\emptyset]\alpha\leftrightarrow\alpha\). Suppose \(XR_\emptyset Y\). If \(\alpha\in X\), then \([\emptyset]\alpha\in X\), so \(\alpha\in Y\) by the definition of \(R_\emptyset\). Thus \(X\subseteq Y\). Since \(R_\emptyset\) is symmetric by the S5 step, the same argument gives \(Y\subseteq X\). Therefore \(X=Y\), so \(R_\emptyset=\mathrm{Id}_W\).

\emph{Inclusion.} If \(C\subseteq D\), the dual of \(\mathsf M_{CD}\) is \([D]\alpha\to[C]\alpha\). Let \(XR_CY\) and assume \([D]\alpha\in X\). Then \([C]\alpha\in X\), hence \(\alpha\in Y\). This is exactly \(XR_DY\). Thus \(R_C\subseteq R_D\).

\emph{Forward composition.} Let \(XR_CY\) and \(YR_DZ\). To show \(XR_{C\cup D}Z\), take any \([C\cup D]\alpha\in X\). The box form of \(\mathsf U_{CD}\) yields \([C][D]\alpha\in X\). From \(XR_CY\) we get \([D]\alpha\in Y\), and from \(YR_DZ\) we get \(\alpha\in Z\). Hence \(R_C\circ R_D\subseteq R_{C\cup D}\).

\emph{Reverse composition.} Assume \(XR_{C\cup D}Z\). We must construct a midpoint \(Y\) such that \(XR_CY\) and \(YR_DZ\). The cases in which one label is empty are forced by identity. If \(C=\emptyset\), take \(Y=X\); then \(XR_\emptyset X\) and, because \(C\cup D=D\), \(XR_DZ\). If \(D=\emptyset\), take \(Y=Z\); then \(XR_CZ\) and \(ZR_\emptyset Z\).

Assume now that \(C,D\neq\emptyset\). The desired midpoint must contain every formula forced by the \([C]\)-theory of \(X\), and it must be compatible with sending its \([D]\)-theory into \(Z\). Define
\[
   \Gamma=
   \{\theta:[C]\theta\in X\}
   \cup
   \{\langle D\rangle\eta: \eta\in Z\}.
\]
The second half of \(\Gamma\) is the canonical way to protect the future \(R_D\)-successor condition: if a candidate midpoint contained \([D]\beta\) while \(\beta\notin Z\), then \(Z\) would contain \(\neg\beta\), so the required diamond \(\langle D\rangle\neg\beta\) would create a contradiction.

We prove that \(\Gamma\) is consistent in a way that exposes exactly where \(\mathsf U_{CD}\) is used. Suppose, for contradiction, that \(\Gamma\) is inconsistent. Then there are finitely many formulas \(\theta_1,\ldots,\theta_r\) with \([C]\theta_i\in X\) and finitely many formulas \(\eta_1,\ldots,\eta_s\in Z\) such that
\[
   \vdash
   (\theta_1\wedge\cdots\wedge\theta_r)
   \to
   \neg(\langle D\rangle\eta_1\wedge\cdots\wedge\langle D\rangle\eta_s).
\]
Put \(\theta=\bigwedge_i\theta_i\), and let empty conjunctions be \(\top\). If \(s=0\), the displayed theorem is \(\vdash\neg\theta\). Since \([C]\theta\in X\) follows from the formulas \([C]\theta_i\in X\) by normal modal reasoning, this contradicts the consistency of \(X\). Hence \(s>0\). Put \(\eta=\bigwedge_{j=1}^s\eta_j\). Because each \(\eta_j\in Z\) and \(Z\) is maximal consistent, \(\eta\in Z\).

The formula \(\eta\to\eta_j\) is a theorem for each \(j\). By necessitation and \(\mathsf K_D\), \([D](\eta\to\eta_j)\) gives \(\langle D\rangle\eta\to\langle D\rangle\eta_j\). Therefore
\[
   \vdash \langle D\rangle\eta
   \to
   \bigwedge_{j=1}^s\langle D\rangle\eta_j .
\]
Combining this theorem with the finite inconsistency display yields \(\vdash\theta\to[D]\neg\eta\). Since \([C]\theta\in X\), normality gives \([C][D]\neg\eta\in X\). Lemma~\ref{lem:dual-forms} converts this to
\[
   [C\cup D]\neg\eta\in X.
\]
Now use the assumed edge \(XR_{C\cup D}Z\): by the definition of the canonical relation, \(\neg\eta\in Z\). This contradicts \(\eta\in Z\). Thus \(\Gamma\) is consistent.

Extend \(\Gamma\) to a maximal consistent set \(Y\). If \([C]\alpha\in X\), then \(\alpha\in\Gamma\subseteq Y\), so \(XR_CY\). To show \(YR_DZ\), let \([D]\beta\in Y\). If \(\beta\notin Z\), then \(\neg\beta\in Z\), so \(\langle D\rangle\neg\beta\in\Gamma\subseteq Y\), contradicting \([D]\beta\in Y\). Hence \(\beta\in Z\). Thus \(YR_DZ\), and \(R_{C\cup D}\subseteq R_C\circ R_D\).

Combining the two inclusions gives \(R_C\circ R_D=R_{C\cup D}\). Finally, if \(XR_CY\), then \([C]\chi\in X\) implies \([C][C]\chi\in X\) by \(\mathsf 4_C\), hence \([C]\chi\in Y\). The reverse direction follows by symmetry of \(R_C\). Thus \([C]\)-theories are constant on \(R_C\)-classes.
\end{proof}

\paragraph{Finite certificates and canonical completeness.}
The completeness theorem below is proved by the canonical model. Section~\ref{sec:complexity} studies a separate finite task in which a candidate \(\mathsf{Dev}(N)\)-model is submitted explicitly. In that task the verifier scans the final relation tables, including the exact equation \(F_C\circ F_D=F_{C\cup D}\), rather than constructing a quotient model.

\begin{theorem}[Soundness and completeness]
\label{thm:completeness}
For the pure coalition-modal language over the fixed finite set of coalition labels \(2^N\), the following are equivalent for every formula \(\varphi\):
\begin{enumerate}[label=(\roman*)]
  \item \(\mathsf H_{\mathrm{bp}}\vdash\varphi\);
  \item \(\varphi\) is valid on every frame in \(\mathsf{Dev}(N)\).
\end{enumerate}
Equivalently, every \(\mathsf H_{\mathrm{bp}}\)-consistent pure modal formula is satisfiable in the canonical \(\mathsf{Dev}(N)\)-model. Social-choice atoms are handled later by choosing a valuation on concrete biprofile models.
\end{theorem}

\begin{proof}
Soundness is checked clause by clause. The schemes \(\mathsf K_C,\mathsf T_C,\mathsf 4_C,\mathsf 5_C\) are valid because each \(E_C\) is an equivalence relation. The scheme \(\mathsf I_\emptyset\) is valid because \(E_\emptyset\) is identity. If \(C\subseteq D\), then \(E_C\subseteq E_D\), so \(\mathsf M_{CD}\) is valid. Finally, \(E_C\circ E_D=E_{C\cup D}\) gives both directions of \(\mathsf U_{CD}\).

For completeness, use the canonical model of \(\mathsf H_{\mathrm{bp}}\). Lemma~\ref{lem:canonical-factorization} proves the required frame facts: the canonical relations are equivalences, \(R_\emptyset\) is identity, they are monotone in the coalition label, and their compositions factor exactly by union. Thus the canonical frame belongs to \(\mathsf{Dev}(N)\). If \(\varphi\) is not derivable, then \(\neg\varphi\) extends to a maximal consistent set \(X\). The canonical truth lemma gives \(\mathfrak M^c,X\not\models\varphi\). Therefore validity on \(\mathsf{Dev}(N)\) implies derivability.
\end{proof}

\section{Manipulation as modal definability}

Sections~4--7 are a translation and audit layer. The point of the translation is not only to rename classical social-choice notions, but to make the moving parts of a strategic event separately checkable. A manipulation claim becomes a modal formula; a satisfying instance becomes a typed edge record \((R,P,C,Q,x,y)\); and each later model change is evaluated by asking which field of that record changed.

The audit map used below is as follows. Strategy-proofness becomes the absence of singleton formulas \(\mu_i\), and group strategy-proofness becomes the absence of \(\mu_C\) for all non-empty coalitions. The field \(R\) stores the welfare comparison; \(P,Q\) store the current and deviated reports; \(C\) stores the allowed coordinates of change; and \(x,y\) store the two rule values. A domain restriction normally deletes the admissibility of \(Q\) or the edge \((R,P)E_C(R,Q)\). An off-domain extension can add a new admissible report and a new rule-table row. A public deletion can leave the displayed endpoint edge in place while destroying the factorization conditions needed for a \(\mathsf{Dev}(N)\)-frame. This is the verification advantage of the true/report language.

For compactness, write
\[
   y\succ_i x := p^i_{yx},
   \qquad
   y\succeq_i x := (y=x)\vee p^i_{yx}.
\]
For a non-empty coalition \(C\), define
\[
  \mathrm{Imp}_C(y,x)=
  \Big(\bigwedge_{i\in C} y\succeq_i x\Big)
  \wedge
  \Big(\bigvee_{i\in C} y\succ_i x\Big).
\]
The following formula separates three roles. First, \(o_x\) records the current outcome. Second, \(\mathrm{Imp}_C(y,x)\) evaluates, using the unchanged true profile \(R\), whether \(y\) is a weak Pareto improvement for the deviating coalition and a strict improvement for at least one of its members. Third, \(\langle C\rangle o_y\) asks whether that outcome is reachable by changing only the coalition's reports.

The coalitional manipulation formula is
\[
  \mu_C=
  \bigvee_{x,y\in A}
  \big(o_x\wedge \mathrm{Imp}_C(y,x)\wedge\langle C\rangle o_y\big).
\]
For singleton coalitions write \(\mu_i\) for \(\mu_{\{i\}}\).

\begin{remark}[Operational reading]
The modality \(\langle C\rangle\) may be read as existential quantification over admissible replacements of the report coordinates in \(C\), with the true profile fixed. Formally,
\[
   \langle C\rangle o_y
   \quad\text{means}\quad
   \exists q_C\in U^C\; f(P[C:=q_C])=y,
\]
where \(U\) is the relevant report domain. This remark is only an operational gloss on the Kripke clause; the proofs use the relations \(E_C\).
\end{remark}

\begin{definition}[Strategy-proofness]
A resolute social choice function \(f\) is \emph{strategy-proof} if, for every true profile \(R\), agent \(i\), and report \(Q_i\),
\[
  f(R)\succeq_i^R f(Q_i,R_{-i}).
\]
It is \emph{strongly group strategy-proof} if there are no non-empty \(C\subseteq N\), true profile \(R\), and joint report \(Q_C\) such that, with \(x=f(R)\) and \(y=f(Q_C,R_{-C})\),
\[
  y\succeq_i^R x\ \text{ for all }i\in C,
  \qquad
  y\succ_j^R x\ \text{ for some }j\in C .
\]
\end{definition}

\paragraph{Preference-domain and group convention.}
Unless explicitly stated otherwise, the paper works with strict linear orders. The weak symbols \(\succeq_i^R\) used for coalitional deviations are induced by strict orders by adding equality of alternatives: \(y\succeq_i^R x\) means \(y=x\) or \(y\succ_i^R x\). We use ``strong group strategy-proofness'' for the no-weak-Pareto-improving-deviation convention: every member of the deviating coalition is weakly better off and at least one member is strictly better off. If an all-strict convention is desired, only the improvement atom \(\mathrm{Imp}_C\) changes.

\begin{lemma}[Deviation witness]\label{lem:witness}
Let \(C\ne\varnothing\) and let \(s=(R,R)\) be sincere. Then
\[
   \mathfrak M_f,s\models \mu_C
\]
if and only if coalition \(C\) has a joint misreport from \(R\) that weakly benefits every member of \(C\) and strictly benefits at least one member.
\end{lemma}

\begin{proof}
If \(\mu_C\) holds, then for some \(x,y\in A\), \(o_x\) holds at \((R,R)\), \(\mathrm{Imp}_C(y,x)\) holds at \((R,R)\), and \(\langle C\rangle o_y\) holds. Hence \(x=f(R)\), all members of \(C\) weakly prefer \(y\) to \(x\) according to \(R\), one member strictly prefers \(y\), and some report profile \(Q\equiv_{-C}R\) satisfies \(f(Q)=y\). This is exactly a profitable joint misreport. The converse reads the same conditions from the definition of such a misreport into the three conjuncts of \(\mu_C\).
\end{proof}

\begin{theorem}[Exact modal characterization]\label{thm:exact}
For every finite resolute social choice function \(f\):
\begin{enumerate}[label=(\roman*)]
  \item \(f\) is strategy-proof iff
  \[
     \mathfrak M_f,(R,R)\models\bigwedge_{i\in N}\neg\mu_i
  \]
  for every \(R\in\mathcal L(A)^N\).
  \item \(f\) is strongly group strategy-proof iff
  \[
  \mathfrak M_f,(R,R)\models
  \bigwedge_{\varnothing\ne C\subseteq N}\neg\mu_C
  \]
  for every \(R\in\mathcal L(A)^N\).
\end{enumerate}
\end{theorem}

\begin{proof}
Part (i) is Lemma~\ref{lem:witness} for singleton coalitions. Part (ii) is the same lemma for all non-empty coalitions.
\end{proof}

\paragraph{Reusable witness object.}
A successful manipulation is stored as
\[
   W=(R,P,C,Q,x,y),
\]
where \((R,P)E_C(R,Q)\), \(x=f(P)\), \(y=f(Q)\), and \(\mathrm{Imp}_C(y,x)\) is evaluated at \(R\). Domain restriction changes whether \(Q\) is an admissible successor, off-domain extension can add the corresponding rule row, and public deletion can preserve the endpoints while destroying factor closure of the surrounding fibre. The record is therefore an audit object as well as a truth witness.

\section{Modal translations of social-choice axioms and the Gibbard--Satterthwaite boundary}

Define
\[
  \mathrm{Par}=
  \bigwedge_{x,y\in A}
  \left(
    \left(\bigwedge_{i\in N}y\succ_i x\right)
    \to \neg o_x
  \right),
\]
which says that no unanimously dominated alternative is chosen. Surjectivity is expressed by
\[
  \mathrm{Onto}=\bigwedge_{x\in A}\langle N\rangle o_x.
\]
For dictatorship, use the point formula
\[
  \mathrm{Dict}_i=\bigwedge_{x\in A}(t^i_x\to o_x).
\]
It says that, at the current sincere profile, the chosen alternative is agent \(i\)'s true top. Global dictatorship is the meta-level condition that the same \(i\) satisfies this formula at every sincere profile.

\begin{proposition}[Correctness of the basic translations]\label{prop:basic}
Let \(f:\mathcal L(A)^N\to A\) be finite and resolute.
\begin{enumerate}[label=(\roman*)]
  \item \(f\) is Pareto efficient iff \(\mathfrak M_f,(R,R)\models\mathrm{Par}\) for every \(R\in\mathcal L(A)^N\).
  \item \(f\) is surjective iff \(\mathfrak M_f,s\models\mathrm{Onto}\) for every state \(s\in S_f\).
  \item Agent \(i\) is a dictator for \(f\) iff \(\mathfrak M_f,(R,R)\models\mathrm{Dict}_i\) for every \(R\in\mathcal L(A)^N\).
\end{enumerate}
\end{proposition}

\begin{proof}
At sincere states the true and reported profiles coincide, so \(\mathrm{Par}\) and \(\mathrm{Dict}_i\) reproduce the usual definitions. For \(\mathrm{Onto}\), the grand coalition can change every reported order while the true profile is irrelevant to the outcome. Hence \(\langle N\rangle o_x\) is true at every state exactly when some report profile selects \(x\).
\end{proof}

\begin{theorem}[Coordinate-explicit translation of Gibbard--Satterthwaite]\label{thm:gs}
Assume \(|A|\geq 3\) and the universal strict-order domain \(\mathcal L(A)^N\). There is no finite resolute social choice function \(f\) satisfying all three conditions below:
\begin{enumerate}[label=(\roman*),leftmargin=8mm]
  \item for every sincere state \((R,R)\),
  \[
     \mathfrak M_f,(R,R)\models \bigwedge_{i\in N}\neg\mu_i ;
  \]
  \item for every state \(s\in S_f\), \(\mathfrak M_f,s\models\mathrm{Onto}\);
  \item for every agent \(i\in N\) there is a sincere state \((R_i,R_i)\) such that
  \[
     \mathfrak M_f,(R_i,R_i)\models \neg\mathrm{Dict}_i .
  \]
\end{enumerate}
Equivalently, every onto resolute rule satisfying the individual no-manipulation scheme at all sincere states has a fixed dictator \(i\), meaning \(\mathfrak M_f,(R,R)\models\mathrm{Dict}_i\) for every \(R\in\mathcal L(A)^N\).
\end{theorem}

\begin{proof}
By Theorem~\ref{thm:exact}, condition (i) is strategy-proofness. By Proposition~\ref{prop:basic}, condition (ii) is surjectivity. Condition (iii) is the modal spelling of non-dictatorship at the meta level: it rules out each fixed agent \(i\) being top-selected at every sincere profile. The finite universal-domain Gibbard--Satterthwaite theorem then excludes the simultaneous satisfaction of (i)--(iii) \cite{Gibbard1973,Satterthwaite1975}.
\end{proof}

\noindent Theorem~\ref{thm:gs} fixes the interface between the classical theorem and the true/report semantics. The report-changing edge \(\langle i\rangle\), the outcome atom, and the true-preference comparison are separate components of the same state. This separation is used below for two audit results that are not statements of the classical theorem itself: new witnesses created by off-domain rows are characterized by Theorem~\ref{thm:boundary-safety}, and surviving witnesses under public deletion are governed by the factor-closure condition of Proposition~\ref{prop:update-safety}.

The domain assumption is mathematically active. Let \(\mathcal D\subseteq\mathcal L(A)^N\) be a set of admissible true profiles and let \(\mathcal M\subseteq\mathcal L(A)^N\) be a set of admissible report profiles, with \(\mathcal D\subseteq\mathcal M\) when truthful reporting is admissible. For \(f:\mathcal M\to A\), define
\[
   \mathfrak M_f^{\mathcal D,\mathcal M}
   =\big(\mathcal D\times\mathcal M,
      \{E_C^{\mathcal D,\mathcal M}\}_{C\subseteq N},V_f\big),
\]
where
\[
 (R,P)E_C^{\mathcal D,\mathcal M}(R',Q)
 \quad\text{iff}\quad
 R'=R,
 \quad Q\in\mathcal M,
 \quad Q\equiv_{-C}P .
\]

\paragraph{Restricted domains and the \(\mathsf{Dev}(N)\) laws.}
The truth clauses for manipulation formulas make sense for any displayed restricted graph \(\mathfrak M_f^{\mathcal D,\mathcal M}\). The completeness theorem, however, applies only to frames satisfying the \(\mathsf{Dev}(N)\) equations. If \(\mathcal M\) is a genuine coordinate product, such as \(\prod_i M_i\), then the restricted relations still satisfy exact union-composition. If \(\mathcal M\) is an arbitrary subset of \(\mathcal L(A)^N\), a mixed report needed to factor an \(E_{C\cup D}\)-edge may be absent. In that case the structure is still a useful labelled audit graph, but it should not be treated as a \(\mathsf{Dev}(N)\)-frame unless the factor-closure condition of Proposition~\ref{prop:update-safety} holds. This is the point at which restricted-domain model checking and \(\mathsf{Dev}(N)\)-validity part ways.

\begin{proposition}[Domain restriction as coordinate restriction]\label{prop:domain}
Let \(f:\mathcal M\to A\) be resolute and let \(\mathcal D\subseteq\mathcal M\subseteq\mathcal L(A)^N\). Then:
\begin{enumerate}[label=(\roman*)]
  \item \(f\) is strategy-proof on \((\mathcal D,\mathcal M)\) iff
  \[
  \mathfrak M_f^{\mathcal D,\mathcal M},(R,R)
  \models \bigwedge_{i\in N}\neg\mu_i
  \]
  for every \(R\in\mathcal D\).
  \item \(f\) is strongly group strategy-proof on \((\mathcal D,\mathcal M)\) iff
  \[
  \mathfrak M_f^{\mathcal D,\mathcal M},(R,R)
  \models
  \bigwedge_{\varnothing\ne C\subseteq N}\neg\mu_C
  \]
  for every \(R\in\mathcal D\).
\end{enumerate}
\end{proposition}

\begin{proof}
The proof of Lemma~\ref{lem:witness} is unchanged except that true profiles range over \(\mathcal D\) and admissible deviations range over \(\mathcal M\).
\end{proof}

The next result is the main social-choice audit theorem used by the examples. It says that, once a rule is known on a restricted report domain, extending it to a larger report domain creates no hidden kind of manipulation: every new failure is witnessed by a boundary row, namely by a deviated report that was not previously admissible.

\begin{theorem}[Boundary-row criterion for off-domain extensions]\label{thm:boundary-safety}
Let \(\mathcal D\subseteq\mathcal M\subseteq\mathcal M'\subseteq\mathcal L(A)^N\), let \(f:\mathcal M\to A\), and let \(g:\mathcal M'\to A\) extend \(f\). Fix a non-empty coalition \(C\). Then \(g\) has a \(C\)-manipulation witness on \((\mathcal D,\mathcal M')\) which is not already a witness for \(f\) on \((\mathcal D,\mathcal M)\) iff there are
\[
   R\in\mathcal D,
   \qquad Q\in\mathcal M'\setminus\mathcal M,
   \qquad Q\equiv_{-C}R,
\]
and alternatives \(x,y\in A\) such that
\[
   x=f(R)=g(R),\qquad y=g(Q),\qquad \mathrm{Imp}_C(y,x) \text{ holds at } R.
\]
Consequently, if \(f\) is strongly group strategy-proof on \((\mathcal D,\mathcal M)\), then \(g\) is strongly group strategy-proof on \((\mathcal D,\mathcal M')\) iff no such boundary-row witness exists for any non-empty \(C\subseteq N\).
\end{theorem}

\begin{proof}
Suppose first that \(g\) has a new \(C\)-manipulation witness at the sincere state \((R,R)\) in the larger model. By Lemma~\ref{lem:witness}, there is a report \(Q\in\mathcal M'\) with \(Q\equiv_{-C}R\), current outcome \(x=g(R)\), deviated outcome \(y=g(Q)\), and \(\mathrm{Imp}_C(y,x)\) true at \(R\). Since \(R\in\mathcal D\subseteq\mathcal M\) and \(g\) extends \(f\), we have \(g(R)=f(R)\). If \(Q\in\mathcal M\), then the same tuple is already an \(f\)-witness on \((\mathcal D,\mathcal M)\), because \(g(Q)=f(Q)\). Since the witness is new, \(Q\notin\mathcal M\). This gives the displayed boundary row.

Conversely, assume that such \(R,Q,x,y\) exist. In \(\mathfrak M_g^{\mathcal D,\mathcal M'}\), the state \((R,R)\) has outcome \(x\), the report \(Q\) is an admissible \(E_C\)-successor because \(Q\equiv_{-C}R\), and the successor has outcome \(y\). The improvement condition is evaluated at the unchanged true profile \(R\). Thus Lemma~\ref{lem:witness} gives a \(C\)-manipulation witness for \(g\). Because \(Q\notin\mathcal M\), this witness could not have existed in the restricted model for \(f\).

The final sentence follows by quantifying over all non-empty coalitions and observing that every larger-domain witness is either an old on-domain witness or one of the boundary-row witnesses just characterized.
\end{proof}

\noindent The theorem is a preservation/localization result, not a new form of the Gibbard--Satterthwaite or Moulin theorems. Those classical results speak about manipulation on a fixed domain. The boundary-row theorem assumes that the old domain has already been certified and asks a different question: which newly admitted report rows can change the certification status? Its answer is exact. Old rows cannot create new witnesses because the rule agrees there; rows outside \(\mathcal M'\) are not admissible; hence every new witness must be located at the boundary \(\mathcal M'\setminus\mathcal M\), with the welfare and outcome fields recorded in the same witness tuple. This is the operational value used in the case study below.

\begin{corollary}[Boundary audit for a proposed extension]\label{cor:boundary-audit}
Assume that \(f\) is strongly group strategy-proof on \((\mathcal D,\mathcal M)\). To verify that an extension \(g:\mathcal M'\to A\) remains strongly group strategy-proof on \((\mathcal D,\mathcal M')\), it is sufficient and necessary to check only rows \(Q\in\mathcal M'\setminus\mathcal M\) that are one-coalition deviations from some sincere \(R\in\mathcal D\). Each failing row returns a unique audit record of the form
\[
   (R,R,C,Q,f(R),g(Q)).
\]
Thus the extension audit is a boundary-row scan rather than a re-analysis of the classical theorem on the larger domain.
\end{corollary}

\begin{proof}
Theorem~\ref{thm:boundary-safety} says that every new larger-domain witness is exactly such a boundary row, and that every such boundary row is a witness. Since \(f\) has no old-domain group manipulation, no other row can create a new failure.
\end{proof}

\section{Single-peaked projections and restricted-domain certificates}\label{sec:single-peaked}

The universal-domain theorem in Section~4 should not be read as a theorem about all natural domains. Fix a linear axis
\[
   \alpha=(a_1\prec_\alpha a_2\prec_\alpha\cdots\prec_\alpha a_m)
\]
on \(A\). A strict order \(R_i\) is single-peaked with respect to \(\alpha\) if, whenever \(a_j\prec_\alpha a_k\prec_\alpha a_\ell\),
\[
   a_j\succ_i a_k \Rightarrow a_k\succ_i a_\ell,
   \qquad
   a_\ell\succ_i a_k \Rightarrow a_k\succ_i a_j .
\]
Let \(\mathsf{SP}_\alpha(A)\) denote the set of such orders.

\begin{proposition}[Size and boundary of the single-peaked coordinate]\label{prop:sp-size}
For a fixed axis \(\alpha\) on \(m\) alternatives,
\[
   |\mathsf{SP}_\alpha(A)|=2^{m-1}.
\]
Hence the restricted biprofile state space
\(\mathsf{SP}_\alpha(A)^N\times\mathsf{SP}_\alpha(A)^N\) has size
\[
   2^{2n(m-1)},
\]
whereas the universal strict-order biprofile space has size \((m!)^{2n}\).
\end{proposition}

\begin{proof}
A single-peaked order can be constructed from worst to best by repeatedly choosing one of the two remaining endpoint alternatives on the axis as the next worst alternative. There are two choices at each of the first \(m-1\) stages and one alternative left at the final stage, giving \(2^{m-1}\) orders. The state-space counts follow by taking \(N\)-fold products for the true and report coordinates.
\end{proof}

The endpoint-deletion coding also gives a computational reading. In the unrestricted domain, a bit string for a reported order must be checked against the axioms of a strict linear order. On a fixed single-peaked axis, by contrast, the \(m-1\) endpoint choices themselves generate exactly the admissible orders. Hence the domain restriction is visible both semantically, as a smaller Kripke coordinate, and computationally, as a simpler report encoding.

\paragraph{State-space reduction view.}
From a search perspective, the universal report generator branches over all permutations of \(A\). The single-peaked report generator branches only over left/right endpoint choices along the fixed axis. Thus the domain restriction does not merely add a filter after a large set of candidates has been generated; it prunes the state tree at generation time. For a coalition \(C\), the children produced by one deviation step fall from \((m!)^{|C|}\) to \(2^{(m-1)|C|}\), as the preceding count shows.

\begin{quote}
\small
\noindent\textsc{GenerateSinglePeakedReport}\((\alpha)\)
\begin{enumerate}[leftmargin=*,label=\arabic*.]
  \item Initialize \(S\leftarrow [a_1,\ldots,a_m]\) and \(order\leftarrow\emptyset\).
  \item For each \(t=1,\ldots,m-1\), choose a bit \(b_t\in\{0,1\}\).
  \item If \(b_t=0\), remove the left endpoint of \(S\) as the next worst alternative.
  \item If \(b_t=1\), remove the right endpoint of \(S\) as the next worst alternative.
  \item Append the remaining alternative as best and return the resulting order.
\end{enumerate}
\end{quote}
Every \((m-1)\)-bit string produces one admissible order and every admissible single-peaked order is produced once. A verifier can therefore quantify over endpoint bits directly, rather than over a permutation code together with a validity circuit.

Proposition~\ref{prop:sp-size} identifies the coordinate at which both the impossibility theorem and the finite checking question change. Taking
\[
   \mathcal D=\mathcal M=\mathsf{SP}_\alpha(A)^N
\]
leaves the manipulation formula \(\mu_C\) and Lemma~\ref{lem:witness} intact, but replaces the universal report quantifier by a single-peaked report quantifier. A direct enumeration over this restricted coordinate gives individual strategy-proofness checking in
\[
   O\big(n2^{(n+1)(m-1)}\big)
\]
time, and strong group strategy-proofness in
\[
   O\big(n2^{n(m-1)}(1+2^{m-1})^n\big).
\]
The corresponding universal-domain bounds are \(O(n(m!)^{n+1})\) and \(O(n(m!)^n(1+m!)^n)\). Thus the single-peaked projection is not only a social-choice assumption; it is a formal reduction of the report-coordinate search space and of the number of existential variables in the symbolic \(\langle C\rangle\)-step. Moulin's theorem on strategy-proofness and single-peakedness is represented by the same no-manipulation formula evaluated on this smaller coordinate product \cite{Moulin1980}.

Let \(\mathrm{peak}_\alpha(R_i)\) be the top alternative of a single-peaked order. For odd \(n\), define the median rule
\[
   f_{\mathrm{med}}(P)=
   \text{the median of }\{\mathrm{peak}_\alpha(P_i):i\in N\}
   \text{ along }\alpha .
\]

For a concrete audit example, take five agents on the axis \(a<b<c<d\) and sincere peaks
\[
   (a,b,b,c,d).
\]
The median outcome is \(b\). Inside the single-peaked report product, the median rule has no profitable singleton deviation for the voter whose true peak is \(d\): the available report changes are still generated by endpoint-deletion orders, and the rule remains the median of reported peaks.

Now keep all single-peaked rule values fixed but inspect an off-domain extension of the rule table. Suppose the extension adds an admissible report
\[
   q_5=d\succ b\succ c\succ a
\]
for agent 5 and stipulates, at the report profile \(Q=(R_1,R_2,R_3,R_4,q_5)\), that the outcome is \(c\). The on-domain row \(P=(R_1,\ldots,R_5)\) still has outcome \(b\). The audit trace is
\[
   W=(R,P,\{5\},Q,b,c).
\]
For the true single-peaked order of agent 5 with peak \(d\), the outcome \(c\) is preferred to \(b\). The same formula \(\mu_5\) is evaluated in both models; what changes is the existence of the labelled successor \((R,P)E_{\{5\}}(R,Q)\). The trace therefore localizes the failure to an added off-domain edge and an added rule-table row, not to a change in the welfare comparison or to the median rule on its intended domain.

The same stored trace also gives an update-safety test. If a public report restriction deletes some states but keeps \(P\) and \(Q\), the manipulation edge may survive while the surrounding fibre may stop satisfying exact union-composition. For two coalitions \(C,D\), the checker must still be able to factor every surviving \(E_{C\cup D}\)-edge through a surviving intermediate state. If the mixed report is missing, the updated transition system can be useful operationally, but it is no longer a \(\mathsf{Dev}(N)\)-frame. This is the concrete diagnostic role of the update-safety check.

\paragraph{Complete replay of the audit.}
The same tuple \(W=(R,P,\{5\},Q,b,c)\) is replayed against three model descriptions.
\begin{quote}
\small
\begin{itemize}[leftmargin=6mm]
  \item Input domains: \(\mathcal D=\mathcal M=\mathsf{SP}_\alpha(A)^5\) and \(\mathcal M'=\mathcal L(A)^5\).
  \item Restricted rule: \(f_{\mathrm{med}}:\mathcal M\to A\), with \(f_{\mathrm{med}}(P)=b\).
  \item Extension row: \(g|_{\mathcal M}=f_{\mathrm{med}}\) and \(g(Q)=c\), where \(Q\in\mathcal M'\setminus\mathcal M\).
  \item Restricted replay: \textsc{edge-deleted}; the report-domain field fails because \(Q\notin\mathcal M\).
  \item Boundary replay: \textsc{boundary-witness}; the new row \(Q\) is admissible, has outcome \(c\), and benefits agent 5.
  \item Deletion replay: \textsc{unsafe-update} if the surviving public image is not factor closed.
\end{itemize}
\end{quote}
This replay is small enough to be inspected by hand but large enough to show the algorithmic delta. For this example, each agent has \(2^{4-1}=8\) single-peaked reports and \(4!=24\) strict reports. Hence the intended single-peaked report space has \(8^5=32768\) rows, while the universal report space has \(24^5=7962624\) rows. The boundary-row theorem says that, after the on-domain median rule has been certified, a proposed extension need not recheck the old \(32768\) rows for new witnesses: any new singleton witness for agent 5 must use a row in \(\mathcal M'\setminus\mathcal M\) such as \(Q\). The public-deletion test is different again: it does not search for a new rule row, but for a missing mixed midpoint in the surviving fibre. Thus the example is a three-way audit transcript rather than merely a toy manipulation witness.

The last line can be made concrete in the same five-agent fibre. Let \(q_4=d\succ c\succ b\succ a\), let \(B=P[4:=q_4]\), let \(C=Q=P[5:=q_5]\), and let \(D=P[4:=q_4,5:=q_5]\). If a public deletion keeps the three corners \(P,B,C\) and deletes the mixed corner \(D\), then the endpoints \(B\) and \(C\) still survive and are connected by the grand-coalition relation on coordinates \(\{4,5\}\). However, the ordered factorization through \(E_{\{5\}}\circ E_{\{4\}}\) needs the deleted midpoint \(D\): from \(B\), changing only report 5 reaches \(D\), and from \(D\), changing only report 4 reaches \(C\). Thus the manipulation edge may remain visible while the public image is no longer a \(\mathsf{Dev}(N)\)-frame. This is exactly the factor-closure failure characterized by Proposition~\ref{prop:update-safety}.

The audit value is the field separation. The restricted model fails at admissibility, the off-domain extension succeeds at the boundary row, and the public deletion fails at the surrounding factorization field. The social-choice theorem used on the intended domain is still Moulin's theorem; the added information is where the same stored strategic event changes under representation changes. This is also where the audit layer differs from generic model checking. A model checker can verify a property on a supplied transition system, but it does not by itself say which rows must be rechecked when a rule is extended off domain or which missing midpoint invalidates a public subframe. The boundary-row theorem and factor-closure criterion give those two delta checks directly: new extension rows are the only possible source of new witnesses, and missing factor midpoints are the exact obstruction to preserving \(\mathsf{Dev}(N)\). In this sense the example is not offered as a new voting theorem; it is a small but complete audit scenario. It starts with a certified single-peaked rule, extends the representation by admitting off-axis reports, and then deletes a public corner. The output is not merely true or false: it identifies the changed field of the witness record and the algebraic reason why the updated graph is or is not a \(\mathsf{Dev}(N)\)-frame.

\noindent This is the practical difference from a bare modal rephrasing of strategy-proofness. A proposed extension of the report domain does not require re-auditing all old on-domain rows: Corollary~\ref{cor:boundary-audit} says that any new witness must use a boundary row. A public deletion does not require guessing which modal axiom failed: Proposition~\ref{prop:update-safety} says that the obstruction is precisely the missing factor midpoint. Thus the witness record functions as a diagnostic interface with operational output: \textsc{edge-deleted} means the proposed report is outside the restricted generator, \textsc{boundary-witness} means a new extension row changes the strategic status, and \textsc{unsafe-update} means the transition system may still be usable as an explicit graph but no longer carries the \(\mathsf{Dev}(N)\) algebra needed by the completeness theorem.

\begin{corollary}[Modal median certificate]\label{cor:median}
For odd \(n\), the median rule on \(\mathsf{SP}_\alpha(A)^N\) satisfies
\[
   \mathfrak M_{f_{\mathrm{med}}}^{\mathsf{SP},\mathsf{SP}},(R,R)
   \models
   \bigwedge_{i\in N}\neg\mu_i
\]
for every \(R\in\mathsf{SP}_\alpha(A)^N\). More generally, generalized median voter schemes satisfy the same formula on the corresponding single-peaked report coordinate.
\end{corollary}

\begin{proof}
The median and generalized median voter schemes are strategy-proof on the single-peaked domain by Moulin's characterization \cite{Moulin1980}. Proposition~\ref{prop:domain} translates that statement into the displayed modal formula.
\end{proof}

The report-coordinate restriction is essential. Theorem~\ref{thm:boundary-safety} gives the exact boundary-row test for whether a larger report domain introduces new witnesses. The simpler monotonicity consequence is the following.

\begin{proposition}[Monotonicity under report-domain expansion]\label{prop:report-expansion}
Let \(\mathcal D\subseteq\mathcal M\subseteq\mathcal M'\), and let \(g:\mathcal M'\to A\) extend \(f:\mathcal M\to A\). If \(f\) is manipulable on \((\mathcal D,\mathcal M)\), then \(g\) is manipulable on \((\mathcal D,\mathcal M')\). Equivalently, strategy-proofness on the larger report domain implies strategy-proofness on the smaller report domain, but not conversely.
\end{proposition}

\begin{proof}
A manipulation witness in \(\mathcal M\) consists of \(R\in\mathcal D\), a coalition \(C\), and a report \(Q\in\mathcal M\) with \(Q\equiv_{-C}R\) satisfying the improvement condition. Since \(\mathcal M\subseteq\mathcal M'\) and \(g\) agrees with \(f\) on \(\mathcal M\), the same witness remains available in the larger model. The contrapositive gives the second statement.
\end{proof}

The corollary, the boundary-row theorem, and the expansion proposition separate two phenomena that are often conflated. Restricted-domain strategy-proofness is not a syntactic weakening of \(\mu_i\); it is a change in the topology, cardinality, and symbolic encoding of the admissible report coordinate. A larger report domain is safe exactly when its new boundary rows fail to form the witnesses characterized in Theorem~\ref{thm:boundary-safety}.

\section{Finite audit records and bounded relational checks}\label{sec:complexity}

This section states the finite audit contract used by the witness language. It is deliberately representation-relative. Once a finite rule table, generator, graph, or certificate is supplied, a checker evaluates formulas and audits labelled edges over that displayed object. This follows the explicit-graph viewpoint of finite-state model checking, where the transition system and the specification are both part of the verification instance \cite{ClarkeEmersonSistla1986,BlackburnDeRijkeVenema2001}. The social-choice side is similarly representation-sensitive: computational social choice treats winner determination, manipulation, and restricted domains as algorithmic problems only after the rule and domain representation has been fixed \cite{BartholdiToveyTrick1989,BrandtConitzerEndrissLangProcaccia2016}. No result below is a succinct-input complexity theorem or a satisfiability upper bound for \(\mathsf{Dev}(N)\); those are separate open problems.

\begin{definition}[Audit input modes]\label{def:implementation-modes}
An audit must declare how its finite object is represented. In listed-rule enumeration, a voting rule is given by its full truth table. If \(n=|N|\), \(m=|A|\), and \(L=m!\), then the input already contains \(L^n\) report rows, and the associated biprofile model has \(L^{2n}\) states because true and submitted profiles are both explicit. In restricted-domain generation, the admissible domain is submitted by a generator, such as the endpoint-deletion code for single-peaked orders, rather than by all \(L^n\) rows. The modal formula is unchanged; only the successor generator used for \(\langle C\rangle\) changes. In typed certificate checking, a pure modal formula \(\varphi\) and a finite certificate are supplied. The coalition set \(N\) is fixed in advance, so \(k=2^n\) is a fixed signature constant for the verifier. In explicit graph model checking, a finite labelled graph is supplied directly and the formula is evaluated over that graph. This is a model-checking problem, not a claim that small graphs always exist for satisfiable formulas.
\end{definition}

The fixed-\(N\) assumption is part of the certificate and explicit-graph modes: coalition labels are compiled into the signature before the input is read. If \(N\) varies with the input, then either \(2^n\) labels must be explicitly listed, or a new succinct label language must be introduced. The results below do neither.

\paragraph{Supplementary material for reproduction.}
The supplementary bundle is text based and deliberately small; \path{MANIFEST.md} and \path{checksums.sha256} list the submitted files. It has three independent parts.
\begin{itemize}[leftmargin=7mm]
  \item The executable certificate part consists of \path{certificate_schema.md}, \path{tools/certcheck.py}, three \path{.cert} examples, \path{certificate_expected_outputs.md}, and \path{certcheck_run.log}. The bundled log and the local reproduction log record an accepting one-state certificate, a rejected diamond row, and a rejected union-composition row with the expected exit statuses.
  \item The Lean part consists of \path{lean/DevNCore.lean}, \path{lean/README_Lean.md}, \path{lean_expected_outputs.md}, and \path{lean_run.log}. It uses only Lean~4's core environment and formalizes the two-agent coalition table, the relational midpoint consequence of exact composition, rectangular mixing, factor closure for public restrictions, and the two-state non-product example. The local Lean~4.29.1 check produced no diagnostics.
  \item The Alloy part consists of \path{devn_alloy_audit.als}, \path{devn_alloy_audit_notes.md}, \path{alloy_expected_outputs.md}, \path{tools/run_alloy_cli.sh}, and \path{alloy_cli_run.log}. It gives bounded searches for non-product components, safe and unsafe deletion, one-corner repair, and boundary-row creation of a new witness. The recorded Alloy~6.2.0 CLI run gives the expected SAT/UNSAT statuses in the stated scopes.
\end{itemize}
For a reviewer who wants a minimal reproduction pass, the supplied finite certificates are checked by one command:
\begin{quote}\small
\texttt{bash tools/run\_certcheck.sh}
\end{quote}
The submitted log \texttt{certcheck\_run.log} records the exact transcript: the accepting certificate exits with status 0, while the bad-diamond and bad-union certificates reject with status 1 for the advertised reasons. The Lean companion was locally checked with \texttt{lean lean/DevNCore.lean} under Lean~4.29.1 and produced no diagnostics; the observed transcript is recorded in \texttt{lean\_run.log}. The Alloy companion was locally executed with Alloy~6.2.0 and OpenJDK~21 using \texttt{tools/run\_alloy\_cli.sh}; the observed SAT/UNSAT statuses are recorded in \texttt{alloy\_cli\_run.log}.

The boundary is part of the claim. The executable checker validates supplied finite tables; the Lean file mechanizes finite relational lemmas; the Alloy model searches small relational scopes. The canonical completeness proof itself remains the mathematical proof in Section~\ref{sec:proof}: the supplement does not formalize maximal consistent sets, the Hilbert calculus, Lindenbaum extension, or the truth lemma. It also does not give a general satisfiability algorithm, a complete model checker for succinct rule descriptions, or an empirical benchmark suite. Its role is narrower and reproducible: it validates the displayed finite certificate format and checks the small relational update patterns used by the audit examples. The submitted ancillary directory records, for each artifact, its input schema, command, observed output, and checksum. The intended reading is therefore three-tiered: mathematical theorems are proved in the paper, finite relation lemmas are guarded by Lean and Alloy companions, and supplied certificates are checked by an executable script with recorded input and output.

\subsection{Witness replay pseudocode}

A practical audit separates three objects that are often conflated. The rule oracle computes outcomes. The domain generator decides which reports are admissible. The modal edge oracle turns admissible reports into labelled successors. A failed strategy-proofness check should therefore return more than a Boolean answer: it should return a typed witness and indicate whether the witness depends on a rule value, an added report, a deleted report, or a failure of exact union-composition after an update.

The recommended witness record is
\[
   W=(R,P,C,Q,x,y,\rho),
\]
where \(R\) is the true profile, \(P\) is the current report, \(Q\) is the deviated report, \(C\) is the changing coalition, \(x=f(P)\), \(y=f(Q)\), and \(\rho\) records the input mode. The shorter tuple \((R,P,C,Q,x,y)\) used above is the mathematical witness; \(\rho\) is implementation metadata.

\begin{quote}
\small
\noindent\textsc{ReplayWitness}\((W,U)\), with \(W=(R,P,C,Q,x,y)\)
\begin{enumerate}[leftmargin=*,label=\arabic*.]
  \item If \(R\) is not an admissible true profile under \(U\), return \textsc{welfare-outside-domain}.
  \item If \(P\) or \(Q\) is not an admissible report profile under \(U\), return \textsc{edge-deleted}.
  \item Compute \(x'=f_U(P)\) and \(y'=f_U(Q)\).
  \item If \(x'\ne x\) or \(y'\ne y\), return \textsc{rule-value-changed}.
  \item If \((R,P)E^U_C(R,Q)\) fails, return \textsc{successor-not-admitted}.
  \item If the improvement test at \(R\) fails, return \textsc{welfare-comparison-changed}.
  \item If the surviving fibre fails exact union-composition, return \textsc{unsafe-update}.
  \item Otherwise return \textsc{same-manipulation-witness}.
\end{enumerate}
\end{quote}

For the five-agent single-peaked example of Section~\ref{sec:single-peaked}, the same record is classified in three ways. The restricted single-peaked model deletes the edge because the off-axis report is not generated. The off-domain extension preserves the witness by adding the report and the corresponding rule row. A public deletion update can keep the endpoints but remove a mixed intermediate report, in which case the surviving graph is no longer a \(\mathsf{Dev}(N)\)-frame.

The public-update issue has a precise finite form. If \(F=(S,\{E_C\}_{C\subseteq N})\) is a \(\mathsf{Dev}(N)\)-frame and \(U\subseteq S\), write \(F\upharpoonright U\) for the induced public restriction with relations
\[
   E_C^U=E_C\cap(U\times U).
\]
The restriction automatically preserves equivalence of each label, empty identity, and inclusion. The only fragile row is exact union-composition.

\begin{proposition}[Update-safety criterion for public restrictions]\label{prop:update-safety}
Let \(F\in\mathsf{Dev}(N)\) and let \(U\subseteq S\). Then \(F\upharpoonright U\in\mathsf{Dev}(N)\) if and only if \(U\) is factor closed: for every \(s,t\in U\) and every \(C,D\subseteq N\),
\[
   sE_{C\cup D}t
   \quad\Longrightarrow\quad
   \exists u\in U\,(sE_Cu\text{ and }uE_Dt).
\]
\end{proposition}

\begin{proof}
The forward direction is just exact composition in the restricted frame. If \(F\upharpoonright U\) is a \(\mathsf{Dev}(N)\)-frame and \(sE_{C\cup D}^Ut\), then exact composition in \(F\upharpoonright U\) gives a midpoint \(u\in U\) with \(sE_C^UuE_D^Ut\).

Conversely, assume factor closure. Restrictions of equivalence relations are equivalence relations on \(U\), \(E_\emptyset^U\) is identity on \(U\), and inclusion is inherited from \(F\). For exact composition, the inclusion \(E_C^U\circ E_D^U\subseteq E_{C\cup D}^U\) follows from exact composition in \(F\). For the reverse inclusion, take \(sE_{C\cup D}^Ut\). Factor closure supplies \(u\in U\) with \(sE_CuE_Dt\), hence \(sE_C^UuE_D^Ut\). Thus \(E_C^U\circ E_D^U=E_{C\cup D}^U\) for all \(C,D\).
\end{proof}

This proposition is the formal audit value of the missing-corner examples. A deletion may preserve the displayed manipulation edge \((R,P)E_C(R,Q)\) and the two rule values, but fail to preserve the surrounding factor closure. The resulting structure is still a labelled transition system, but it should not be treated as a public subframe of \(\mathsf{Dev}(N)\).

\subsection{Typed finite certificates}

A typed certificate for a pure coalition-modal formula \(\varphi\) consists of a finite set \(T\) of states, an initial state \(t_0\), a finite closure \(\Delta\) of formulas containing \(\varphi\), final relation tables \(R_C\subseteq T^2\) for every coalition label, a written truth type \(\lambda(t)\subseteq\Delta\) for each state, and optional witness pointers for positive diamonds and union factors. The relation tables are submitted as final tables to be checked; the verifier does not construct them from a quotient or from a canonical model.

The plain input schema used in the supplementary examples is the following record:
\begin{quote}
\small
\begin{verbatim}
states:       finite list of state names
root:         one state name
labels:       fixed coalition labels, including empty and grand labels
formula:      root formula to be checked
closure:      finite formula closure Delta
types:        for each state, the formulas from Delta claimed true there
relations:    for each label C, the listed pairs of R_C
diamonds:     optional witnesses for positive diamond rows
factors:      optional midpoints for positive union-composition rows
\end{verbatim}
\end{quote}

For example, the one-state accepting certificate for \(p\) has state \(s_0\), root \(s_0\), type \(\lambda(s_0)=\{p,[\emptyset]p,\langle\emptyset\rangle p\}\) for the displayed closure, and every relation equal to \(\{(s_0,s_0)\}\). A failed-diamond example claims \(\langle\{1\}\rangle p\) at \(s_0\) while every \(R_{\{1\}}\)-successor omits \(p\). A failed-union example uses two states, singleton relations equal to identity, and a universal grand-coalition relation. The supplement includes these three files and a small line-based reference checker, \texttt{tools/certcheck.py}; its transcript accepts the first certificate and rejects the other two at the diamond and union-composition rows, respectively.

The exact reproduction command is intentionally short:
\begin{quote}
\small
\begin{verbatim}
bash tools/run_certcheck.sh
\end{verbatim}
\end{quote}
The bundled \texttt{certcheck\_run.log} records the corresponding outputs and exit statuses: \texttt{ACCEPT} for the positive example, rejection by modal truth mismatch for the bad-diamond example, and rejection by failed union-composition for the bad-union example. The checker validates only supplied finite tables; it is not a satisfiability solver.

\begin{quote}
\small
\noindent\textsc{VerifyTypedCertificate}\((\mathcal T,\varphi)\)
\begin{enumerate}[leftmargin=*,label=\arabic*.]
  \item Check that each type is a complete Boolean type over \(\Delta\).
  \item Check that every \(R_C\) is an equivalence relation.
  \item Check \(R_\emptyset=\mathrm{Id}\) and \(R_C\subseteq R_D\) whenever \(C\subseteq D\).
  \item For every \(C,D,t,v\), check that \(tR_{C\cup D}v\) holds exactly when some \(u\) satisfies \(tR_CuR_Dv\).
  \item For every \([C]\psi\in\Delta\) and state \(t\), check the displayed box truth row against all \(R_C\)-successors of \(t\).
  \item For every \(\langle C\rangle\psi\in\Delta\) and state \(t\), check the displayed diamond truth row against all \(R_C\)-successors of \(t\).
  \item For each positive diamond, check that the displayed witness pointer names such a successor.
  \item Accept iff all rows pass and \(\varphi\in\lambda(t_0)\).
\end{enumerate}
\end{quote}

The certificate is an exchange format for auditing, not a proof search method. A tool may first find a manipulation witness in a concrete social-choice model, then export only the finite neighbourhood needed to check the modal claim. Another tool can verify the exported relation tables without trusting the original rule program. This is the practical role of the certificate: it decouples witness generation from witness validation.

\begin{proposition}[Polynomial verification of typed certificates]
\label{prop:certificate-verifier}
For fixed \(N\), there is a deterministic polynomial-time verifier, measured in the written certificate size, that decides whether a submitted typed finite certificate is a genuine finite \(\mathsf{Dev}(N)\)-model satisfying \(\varphi\) at \(t_0\).
\end{proposition}

\begin{proof}
The verifier reads the final tables and checks the following finite rows. Boolean type rows check that each \(\lambda(t)\) is coherent over \(\Delta\). Frame rows check that each \(R_C\) is an equivalence relation, that \(R_\emptyset=\{(t,t):t\in T\}\), and that \(C\subseteq D\) implies \(R_C\subseteq R_D\). Union rows check, for every \(t,v,C,D\), that
\[
   tR_{C\cup D}v
   \quad\Longleftrightarrow\quad
   \exists u\in T\,(tR_Cu\text{ and }uR_Dv).
\]
For the left-to-right direction a supplied factor pointer may name a displayed midpoint; for the right-to-left direction the verifier scans all \(u\in T\). Box rows and diamond rows are then checked as truth equivalences against the final relation tables. Positive diamonds may carry witness pointers, but negative diamonds are not trusted as missing-pointer claims: they are checked by scanning all successors. The root row checks \(\varphi\in\lambda(t_0)\).

If all rows pass, the frame rows and union rows say exactly that \((T,\{R_C\})\in\mathsf{Dev}(N)\). The Boolean rows handle propositional connectives, and the modal rows are equivalence checks. An induction on formulas in \(\Delta\) gives
\[
   \mathcal T,t\models\chi
   \quad\text{iff}\quad
   \chi\in\lambda(t)
\]
for every retained formula \(\chi\) and every \(t\in T\). In particular, \(\mathcal T,t_0\models\varphi\).

With \(p=|T|\), \(h=|\Delta|\), and fixed \(k=2^{|N|}\), the union row costs \(O(k^2p^3)\). The modal truth rows cost \(O(kp^2h)\), because each box or diamond formula is checked by scanning the relevant successor row. The remaining type, duality, pointer, and root checks are polynomial in \(p,h\). Thus the submitted finite certificate is efficiently checkable in its written size. This proves only a verification statement: it assumes that the finite tables have been supplied and does not bound the size of certificates needed for arbitrary satisfiable formulas.
\end{proof}

\subsection{Lean proof companion}

The supplementary file \texttt{lean/DevNCore.lean} is a small Lean~4 companion for the finite relational core \cite{MouraUllrich2021Lean4}. Its purpose is deliberately narrower than the canonical proof of Theorem~\ref{thm:completeness}. It does not formalize maximal consistent sets, the Hilbert calculus, Lindenbaum extension, or the canonical truth lemma. Instead, it mechanizes the algebraic obligations that are easiest to mistranscribe in the finite audit layer.

The file fixes the two-agent coalition table with labels \(\emptyset,\{1\},\{2\},\{1,2\}\), defines a \(\mathsf{Dev}\)-frame as labelled relations satisfying equivalence, empty identity, and exact composition, and proves the following relational facts.
\begin{quote}
\small
\begin{verbatim}
midpoint
rectangular_mixing
restricted_composition_iff_factor_closed
nonProductFrame_not_coordinate_separated
\end{verbatim}
\end{quote}
The first theorem is the relational form of the reverse-composition midpoint: from a displayed \(R_{C\cup D}\)-edge one obtains a midpoint for an \(R_C\)-edge followed by an \(R_D\)-edge. The second specializes this to the mixed-corner property. The third proves the composition row for a public restriction under factor closure. The fourth mechanizes the two-state abstract component whose non-empty labels are universal and whose coordinate separation fails.

A reader with Lean~4 can check the companion with
\begin{quote}
\small
\begin{verbatim}
lean lean/DevNCore.lean
\end{verbatim}
\end{quote}
A successful run produces no output. The submitted \texttt{lean\_run.log} records a local Lean~4.29.1 check with no \texttt{sorry}/\texttt{admit} matches and no diagnostics from the Lean command. The Lean file is therefore a proof companion for the finite relational core, while the canonical completeness theorem remains the mathematical proof in Section~\ref{sec:proof}. This split is intentional: it gives a machine-checked guardrail for the audit algebra without turning the paper into a full proof-assistant formalization project.

\subsection{Alloy companion checks}

The supplementary file \texttt{devn\_alloy\_audit.als} is a bounded relational audit model for the finite frame and public-update side conditions \cite{Jackson2012Alloy}. It fixes two agents and represents the four coalition labels \(\emptyset,\{1\},\{2\},\{1,2\}\). The model encodes the \(\mathsf{Dev}(N)\) laws, coordinate separation, rectangular mixing, public deletion, and the factor-closure condition from Proposition~\ref{prop:update-safety}. It is used for bounded counterexample search and example generation: it checks that the intended preservation assertions have no small counterexample in the chosen scopes, and it runs the product, non-product, unsafe deletion, surviving-witness, and one-corner repair scenarios used in the text.

A reader can reproduce the intended checks in the Alloy Analyzer with the following commands.

\begin{quote}
\small
\begin{verbatim}
check DevImpliesRectangularMixing for 6 State, 4 Report
check FactorClosureCharacterizesSafeDeletion for 6 State, 4 Report
run MissingCornerInsideDev for 6 State, 4 Report
run NonProductDevComponent for exactly 2 State, 4 Report
run SafeLineDeletion for exactly 4 State, 4 Report
run UnsafePublicDeletion for exactly 4 State, 4 Report
run UpdateBreaksExactComposition for exactly 4 State, 4 Report
run WitnessSurvivesButUpdateUnsafe for exactly 4 State, 4 Report
run RepairByAddingOneCorner for exactly 4 State, 4 Report
check NewWitnessIsOldOrBoundary for 0 State, exactly 4 Report
run BoundaryRowCreatesNewWitness for 0 State, exactly 4 Report
\end{verbatim}
\end{quote}

The observed Alloy~6.2.0 CLI run in \texttt{alloy\_cli\_run.log} returns \texttt{UNSAT} for the two preservation checks and for \texttt{MissingCornerInsideDev}, \texttt{SAT} for the non-product, safe-deletion, unsafe-deletion, failed-composition, surviving-witness, repair, and boundary-row instance searches, and \texttt{UNSAT} for \texttt{NewWitnessIsOldOrBoundary}. Thus exact union-composition has no bounded counterexample to rectangular mixing in the selected scopes; factor closure has no bounded counterexample to public-update safety; missing corners are not found inside a genuine \(\mathsf{Dev}(N)\)-frame; and the expected finite patterns are generated for non-product components, unsafe public deletion, one-corner repair, and boundary-row witness creation. These bounded checks do more than test syntax: they generate the small update and boundary-row patterns discussed in the social-choice audit layer while leaving the general completeness proof to Section~\ref{sec:proof}.

\subsection{Enumeration baselines and complexity boundary}

\begin{theorem}[Naive enumeration bounds for listed rules]
\label{thm:enumeration-bounds}
Assume that \(f\) is given by a truth table and that preference comparisons are precomputed.
\begin{enumerate}[label=(\roman*)]
  \item Model checking a formula \(\varphi\) at all states by direct enumeration takes time
  \[
     O\bigl(|\varphi|L^{2n+d(\varphi)}\bigr),
  \]
  where \(d(\varphi)=\max\{|C|:\langle C\rangle\psi\text{ occurs in }\varphi\}\).
  \item Strategy-proofness can be checked by direct enumeration in time \(O(nL^{n+1})\).
  \item Strong group strategy-proofness can be checked by direct enumeration in time
  \[
     O\bigl(nL^n(1+L)^n\bigr).
  \]
\end{enumerate}
\end{theorem}

\begin{proof}
For (i), evaluate subformula truth sets bottom-up. Atomic and Boolean clauses are linear in the state space. For a modal clause \(\langle C\rangle\psi\), each state requires enumeration of at most \(L^{|C|}\) report replacements, so the worst modal factor is \(L^{d(\varphi)}\). For (ii), enumerate each true profile, each agent, and each alternative report for that agent. For (iii), enumerate each true profile, each non-empty coalition, and each joint report for that coalition; summing over coalitions gives the displayed bound.
\end{proof}

These are deliberately naive bounds for listed rules. They are useful as baselines, not hardness results and not upper bounds for succinctly represented voting rules. For a restricted-domain generator, replace \(L\) in the successor loop by the number of reports generated for one agent. For the single-peaked endpoint generator of Section~\ref{sec:single-peaked}, the branching factor is \(2^{m-1}\) rather than \(m!\). This is the computational counterpart of the modal point made earlier: the formula is unchanged, but the successor generator changes.

\begin{remark}[Complexity status]
\label{rem:complexity-status}
Propositional satisfiability embeds into the pure modal language, so \(\mathsf{Dev}(N)\)-satisfiability is NP-hard even for fixed \(N\). For ordinary S5, Ladner's theorem gives NP-completeness \cite{Ladner1977}. The present abstract frame class adds \(R_\emptyset=\mathrm{Id}\), monotone inclusion, and exact union-composition, so a matching upper bound or small-model analysis would require a separate filtration argument respecting the coalition table. Proposition~\ref{prop:certificate-verifier} supplies the verification component needed when finite \(\mathsf{Dev}(N)\) certificates are explicitly exchanged.
\end{remark}

For reproducibility, an audit report should state the rule representation, the true-profile and report-profile domains \((D,M)\), the domain generator if one is used, whether coalition labels are fixed in the signature or encoded as input data, whether relations are generated on demand or submitted as explicit final tables, the returned witness \(W\) or finite certificate \(\mathcal T\), and whether the result is a listed-rule enumeration result, a generator-based audit, an explicit-graph model-checking result, or a pure modal certificate-verification result. These fields make the section practically usable while preserving the theoretical boundary: the paper proves completeness for the abstract frame class, correctness of the social-choice translations, and polynomial checking of supplied finite certificates; it leaves the tight satisfiability complexity of \(\mathsf{Dev}(N)\) open.

\section{Conclusion}

Biprofile deviation logic treats strategic deviation as a labelled report-replacement relation while keeping welfare comparisons on a fixed true profile. The main technical result is the completeness of \(H_{\mathrm{bp}}\) for \(\mathsf{Dev}(N)\), whose S5 labels interact by empty-coalition identity, inclusion, and exact union-composition. The proof displays the canonical factorization step needed for that finite semilattice table and connects it to a reusable report-deviation semantics in which proof theory, coordinate recovery, and strategic-audit records share the same labelled relations.

The representation theorem marks the boundary between abstract deviation frames and concrete report products. Exact union-composition already gives finite rectangular mixing inside \(\mathsf{Dev}(N)\); what the abstract equations do not force is coordinate separation. This lets the proof theory work at the level of labelled report-deviation algebra while still identifying the extra condition needed to recover report coordinates.

The social-choice contribution is the audit layer built on that semantics. Classical strategy-proofness, Gibbard--Satterthwaite, and single-peaked median facts are translated into the true/report language; the paper-specific tests then ask what happens when the representation changes. Theorem~\ref{thm:boundary-safety} and Corollary~\ref{cor:boundary-audit} turn off-domain extension checking into a boundary-row scan, while Proposition~\ref{prop:update-safety} identifies factor closure as the exact public-update safety condition. The five-agent median case study shows the intended use: the same stored strategic event is classified as deleted by the restricted generator, created by a boundary row of an off-domain extension, and made algebraically unsafe by a missing public-update midpoint. These results do not strengthen the classical impossibility theorems, but they tell a verifier which witness field changed when a domain is extended or a public submodel is deleted. The finite-audit section connects the tests to reproducible finite objects: replay pseudocode identifies the changed witness field, the bundled checker validates the supplied certificate examples, the Lean companion proves the small relational midpoint and factor-closure lemmas, and the Alloy companion searches the safe and unsafe finite patterns used in the update-safety and boundary-row discussions.

\section*{Author Contributions}
Conceptualization, F.A. and B.B.; methodology, F.A.; formal analysis, F.A.; writing--original draft preparation, F.A.; writing--review and editing, F.A. and B.B.; supervision, B.B. All authors have read and agreed to the published version of the manuscript.

\section*{Funding}
This research received no external funding.

\section*{Data Availability Statement}
No empirical data were created or analyzed in this study. The supplementary material contains the LaTeX sources; \texttt{devn\_alloy\_audit.als} with bounded Alloy outcomes and the recorded \texttt{alloy\_cli\_run.log}; \texttt{lean/DevNCore.lean} with the recorded \texttt{lean\_run.log}; the supplied-certificate schema; \texttt{tools/certcheck.py}; accepting and rejecting \texttt{.cert} examples; \texttt{certcheck\_run.log}; and checksums for the submitted files.

\section*{Institutional Review Board Statement}
Not applicable.

\section*{Informed Consent Statement}
Not applicable.

\section*{Conflicts of Interest}
The authors declare no conflicts of interest.